\documentclass[reqno,11pt,a4paper]{amsart}

\usepackage{pgfplots}
\usepackage{caption}
\usepackage{subcaption}

\usepackage{amsfonts,amsmath,amssymb}
\usepackage[latin1]{inputenc}
\usepackage{dsfont}
\usepackage{enumerate}
\usepackage{xcolor}
\usepackage[all]{xy}

\usepackage{hyphenat}

\usepackage{enumitem}
\usepackage{mathtools}

\usepackage{graphicx}

\usepackage{multirow}
\usepackage{array}

\usepackage{amssymb}
\usepackage{amsfonts}
\usepackage{mathrsfs}
\usepackage{hyperref}

\newcommand{\R}{\mathds R}

\newtheorem{thm}{Theorem}[section]
\newtheorem{prop}[thm]{Proposition}
\newtheorem{lemma}[thm]{Lemma}

\theoremstyle{definition}

\newtheorem{rem}[thm]{Remark}

\newcommand{\ben}{\begin{enumerate}}
	\newcommand{\een}{\end{enumerate}}
\newcommand{\bit}{\begin{itemize}}
	\newcommand{\eit}{\end{itemize}}
\newcommand{\edoc}{\end{document}}

\DeclareMathOperator*{\essinf}{ess\,inf}

\title[On the formulations of the Fermat principle]{On the formulations of the Fermat principle in general relativity and beyond}
\author[E. Caponio]{Erasmo Caponio}\address{Dipartimento di Meccanica, Matematica e Management\hfill\break\indent
		Politecnico di Bari \hfill\break\indent Via Orabona 4,
		70125, Bari, Italy}
	
\email{erasmo.caponio@poliba.it}

\author[M. A. Javaloyes]{Miguel Angel Javaloyes}\address{Departamento de Matem\'aticas \hfill\break\indent Universidad de Murcia \hfill\break\indent Campus de Espinardo, 30100 Espinardo, Murcia, Spain} \email{majava@um.es}

\begin{document}
\begin{abstract}
	This paper presents a survey of the Fermat principle within the framework of general relativity, tracing its evolution from classical optics to its modern variational formulation in Lorentzian geometry. In particular, we provide its proof in the framework of smooth lightlike curves. We also analyze the mathematical difficulties inherent in the relativistic setting, specifically demonstrating that the space of lightlike curves in the Sobolev topology does not admit a $C^1$-manifold structure due to the cone nature of the null condition. To address these variational obstacles, we discuss alternative frameworks highlighting the role of the quadratic arrival time functional in establishing multiplicity results for light rays. Furthermore, we explore significant extensions of the principle, such as its application to extended sources and receivers, arbitrary arrival curves,  timelike geodesics with prescribed proper time,    Finsler spacetimes, or settings with a non-continuous interface giving rise to a Snell law.
	\end{abstract}
	
\maketitle

\begin{center}
\begin{minipage}{0.8\textwidth}
{\itshape
\noindent In celebration of Paolo Piccione's 60th birthday, this work is dedicated to him as a token of gratitude for his inspiring mentorship and friendship over the years}
\end{minipage}
\end{center}

\bigskip
\section{Introduction}
The Fermat principle, a foundational concept in optics, states that light travels between two points along the path of least time. This principle, initially formulated for Euclidean space, finds a profound and elegant extension within the framework of general relativity (GR). In GR, the presence of mass and energy warps spacetime, and this curvature dictates the path of light. The  principle states that light rays, alias lightlike geodesics, follow the path that makes stationary  the proper time of an observer or, equivalently,  the one of a light source. Its history  traces back to the early days of GR when scientists sought to understand how gravity bends light, a phenomenon famously confirmed by Arthur Eddington's 1919 eclipse expedition.

The application of the Fermat principle in GR is crucial for understanding several astrophysical phenomena. A primary example is {\em gravitational lensing}, where a massive object (like a galaxy, a cluster of galaxies, a black hole) acts as a lens, bending the light from a distant source. 
In the weak field, thin-lens approximation, the lensing problem can be encoded by  the so-called Fermat potential on the lens plane; the observable images are precisely its stationary points, while degenerate critical points produce caustics and govern the change of image multiplicities and parities 	\cite{BN86}.  From the spacetime perspective, however, this should be regarded as an approximate realization of the exact general-relativistic Fermat principle, where light rays are characterized as lightlike geodesics making stationary  the proper arrival time of the observer.

The manuscript is organized as follows. We begin with a historical overview of the Fermat principle, from its classical optical origins to its relativistic formulation. Next, we present the variational framework for lightlike curves, discussing the construction of admissible variations and the reduction of the relativistic problem to a more classical variational setting. We then analyze the regularity issues of the corresponding path spaces, emphasizing the obstruction to a Sobolev manifold structure and the role of smoother categories. After that, we explain how these difficulties can be bypassed in relevant geometric settings by means of suitable reductions and quadratic functionals. Finally, we review several extensions of the principle, including the cases of extended sources and receivers, timelike geodesics with prescribed proper time, arbitrary arrival curves, Finsler spacetimes, and discontinuous interfaces leading to Snell-type laws.

\subsection{History of the Fermat principle}
\subsubsection{Classical Fermat principle}
The behaviour of light has been studied since the ancient times. Euclid already established that light travels in a straight line around 300 BC, while Hero of Alexandria,  1st or 2nd century AD, deduced the law of reflection from a principle of minimum distance. The law of refraction was not understood so early. Though Ptolemy found an approximation valid for small angles, it was the Persian scientist Ibn Sahl at Bagdad court in 984 AD who first discovered the law in his study of  lenses  focusing light. But the work by Ibn Sahl was forgotten, and some years later, Ibn al-Haytam (Alhazen), probably the first true scientist of history,  almost rediscovered the law,  establishing  that  the normal to the surface,  the incident  ray  and the refracting  ray   lie in the same plane. 

Later the refraction law was rediscovered several times. First in 1602 by Thomas Harriot, who did not publish his results but informed Kepler by correspondence. Then by the Dutch astronomer Willebrord Snellius, who  also did not publish his result during his lifetime, and later by Ren\'e Descartes in 1637 in his essay Dioptrique. Some years later, in 1657, Pierre de Fermat received from Marin Cureau de la Chambre a copy of his new treatise in  which  he complained about the fact that Hero's principle of minimum distance does not apply to explain refraction of light. Pierre de Fermat gave the solution to this question four years later in a letter to de la Chambre in which he explained that denser mediums create a resistance to light propagation which makes the speed of light slower and that light takes the path of least time. Remarkably, this was some years before  the development of calculus by Newton and Leibniz, but Fermat had already developed a method to calculate minima and maxima,  the method of adequality. 
Initially,  there was reluctance to accept the Fermat principle because it was thought that the principle attributed rationality to nature, in the sense that light had to choose the path of least time among all possible paths from one point to another. Fermat died the same year in  which  Robert Hooke first proposed the wave theory of light in 1665 and he was not alive to witness the new premise of Huygens' principle. It was in 1678 when Christiaan Huygens proposed that every point in the spherical wave front is the source of a new spherical wave front. Huygens was able to prove that his principle could explain the law of refraction, but it seems that he did not realize that it implies also the principle of least time. 
For a detailed historical account of these developments, see \cite{Dar12}.

\subsubsection{Early GR history}

The original classical Fermat principle was formulated in a Euclidean space with time as an independent non-relativistic quantity. But it is not so obvious how to generalize the principle to relativistic spacetimes, where time flow depends on the observer. Soon after Einstein formulated the theory of general relativity, H. Weyl established a version of the Fermat principle for static spacetimes \cite{weyl17}. The advantage of considering static spacetimes is that there is a global notion of time, let us say, a universal time, and also a distinguished notion of rest space. 
We can directly translate the classical Fermat principle to this background, though the destination shifts from a single point to a static worldline. The primary difference is that the coordinate speed varies from point to point, as the metric is not necessarily flat. Consequently, these trajectories are best described as lightlike curves.

The next class of spacetimes in which the Fermat principle was considered (see \cite{LeCi27, syn25}  and also \cite{Pham57,LauLif62}) is that of stationary spacetimes.  The main difference with the static case is that there is no global notion of rest space. But there is still a clear notion of universal time and there are also stationary orbits. This means that one can consider a sort of quotient rest space as the space of stationary wordlines. Then one can generalize the principle by considering lightlike curves from one event to one stationary worldline; light rays follow the trajectories that arrive first at the worldline, which turn out to be lightlike geodesics. Later, in the 1970s, K.~Uhlenbeck further generalized Fermat's principle to globally hyperbolic spacetimes in her study of Morse theory for geodesics on Lorentzian manifolds \cite{Uhlenb75}, highlighting how Fermat's principle is an excellent tool for the study of lightlike geodesics.

Observe that  in Classical Mechanics it was the Fermat principle which determined the trajectories of light rays without being subject to other physical laws. At most, it could be deduced from Huygens' principle, but apart from these two principles and observation, there wasn't a specific theory to describe light rays. Nonetheless, in GR there was  a precise  description of the trajectories of light rays at least in smooth backgrounds. These trajectories are lightlike geodesics, which can be computed solving an ODE.  This means that if one wants to generalize the Fermat principle to arbitrary spacetimes, then one has to prove that the light rays that satisfy the principle are lightlike geodesics. This was a difficult enterprise for many years, first because there isn't a global time in arbitrary spacetimes, and then because even when I. Kovner  \cite{Kov90} observed that one could consider the proper time of the observer receiving the light ray, it wasn't so evident how to prove that lightlike geodesics minimize the proper time of the observer. The suitable variational setup was finally found by V. Perlick \cite{Per90}. The final statement says that light rays from one event (a point from the spacetime) to an observer (a timelike curve) are critical points of the arrival time in the space of lightlike curves of a suitable regularity. Here the arrival time is the parameter of the observer in the point where the lightlike curve touches the observer, which in particular can be chosen as the proper time of the observer. A comprehensive  treatment of  Fermat-type principles in relativity is provided in \cite{Per00}, where one finds the reduction formalism from which the later Finsler description of conformally stationary spacetimes naturally emerges \cite{CJM11}.

\section{The variational formulation}\label{Fermat}

Although the Fermat principle was originally formulated without a formal proof, it accurately describes the trajectories of light rays. It acts as a bridge in geometrical optics, as it can be derived from Huygens' principle and used to deduce Snell's law. In GR, as we already have a description of light rays as lightlike geodesics, at least in smooth settings,  the Fermat principle requires a proof. 

 In this section, we  revisit the proof given by V. Perlick in \cite{Per90}, postponing a functional-analytic proof to Subsection~\ref{Fermatsmooth}.

The setup is the following. Let $(M,g)$ be a Lorentzian manifold, $p\in M$ an arbitrary event and $\alpha:I\subset \R\rightarrow M$ an (injective) timelike curve. Then we consider all the $C^1$-lightlike curves $\gamma:[a,b]\rightarrow M$ such that $\gamma(a)=p$ and $\gamma(b)\in {\rm Im}(\alpha)$ and the arrival time is given by 
\begin{equation}\label{T}
	T(\gamma)=s, \quad \text{ where $s=\alpha^{-1}(\gamma(b))$.} \end{equation}

\subsection{The variational vector fields of variations by lightlike curves} Just considering a variation of a curve $\gamma$ by lightlike curves $\gamma_w$,  with $w\in (-\varepsilon,\varepsilon)\subset \R$,  and deriving with respect to the variational variable $w$, one easily deduces that the variational vector field $Z$ along $\gamma$ satisfies
\begin{equation}\label{vareq}g(\dot\gamma, \frac{D}{dt} Z)=0,
\end{equation}
where $\frac{D}{dt}$ is the covariant derivative along $\gamma$ determined by the Levi-Civita connection of $g$. 
This is because 
\[\frac{\partial}{\partial w}g(\dot\gamma_w,\dot\gamma_w)|_{w=0}=2g(\frac{D}{dw}\dot\gamma_w,\dot\gamma_w)|_{w=0}
=2g(\frac{D}{dt}Z,\dot\gamma).\]
As $g(\dot\gamma_w,\dot\gamma_w)=0$ for all $w\in(-\varepsilon,\varepsilon)$, 
it follows that \eqref{vareq} holds. The next step is to check that \eqref{vareq} is enough to guarantee that $Z$ is the variational vector field of a variation  by  lightlike curves. There are several possibilities to consider the space of lightlike curves in the variation:
\begin{enumerate}
	\item One can consider variations through (piecewise) smooth curves. In this setting, the arrival time is well-defined for each variation, and one can call a curve a critical point if the first variation vanishes for all (piecewise) smooth variations. However, because the breakpoints may occur at arbitrary parameter values, the space of piecewise smooth lightlike curves does not carry a natural manifold structure (just as the space of all piecewise smooth curves does not).

	\item Another natural possibility is to consider a space of lightlike curves with a certain regularity in such a way that this space admits a manifold structure. This is the case of $C^1$ future-pointing lightlike curves (see Subsection~\ref{C1}).  On the contrary, the space of  $H^1$-lightlike curves  does not admit a $C^1$-manifold structure (see Subsection~\ref{W1p}). We will see in Subsection~\ref{bypassing} as  the lack of regularity of the space of  $H^1$-lightlike curves  can be circumvented in some cases.
\end{enumerate}
\subsection{The setting of  smooth curves} V. Perlick worked in \cite{Per90} with smooth lightlike curves and proved  that any smooth vector field $Z$ along  $\gamma$  satisfying
\begin{equation}\label{variationalvf}
	g\!\left(\dot\gamma,\frac{D}{dt}Z\right)=0 \quad \text{on } [a,b],
	\qquad
	Z(a)=0,
	\qquad
	Z(b)\in\mathrm{span}\{\dot\alpha(T(\gamma))\}
\end{equation}
is the variational vector field of a variation by smooth lightlike curves with initial point fixed at $p$ and final point constrained to lie on $\alpha$.   The proof  was accomplished in several steps:
\begin{enumerate}
	\item[\em Step 1] For every $s\in [a,b]$ choose a chart $\psi_s$  defined in an open subset  $\Omega_s=A_s\times [a_s,b_s]\subset \R^n$ such that $\partial_n$ is timelike, $g(\dot\gamma,\partial_n)<0$, and there is an interval $(t^s_0,t^s_1)$ such that $\gamma|_{(t^s_0,t^s_1)}\subset \Omega_s$ and the projection of $\gamma|_{[t^s_0,t^s_1]}$ on $A_s$ is an embedded  curve. 
	We additionaly will require that $\gamma$ is a straight line parametrized with one of the coordinates, contained in $x_n=0$ and that  $\partial_n|_{\gamma(t)}$  does not belong to ${\rm span}\{Z(t),\dot\gamma(t)\}$  (these conditions will be used in {\em step 4*}).
	Moreover,  in this chart the metric $g$ can be expressed as
	\[g=e^{2f}\left( \sum_{i,j=1}^{n-1}\hat g_{ij} dx^i\otimes dx^j- \big(dx^n+\sum_{i=1}^{n-1}\hat\phi_idx^i\big)^2\right) \]
	for certain functions $f, \hat\phi_i$ satisfying that $e^{2f}=-g_{nn}$ and $\hat\phi_i=\frac{-g_{ni}}{g_{nn}}$ and $\hat g_{ij}:=-\frac{g_{ij}}{g_{nn}}+\frac{g_{ni}g_{nj}}{g_{nn}^2}$, where $g_{ij}$ are the coefficients of $g$ in the chart $(\Omega_s,\psi_s)$. Taking into account that $g(\dot\gamma,\dot\gamma)=0$, one can obtain the coordinate $\dot\gamma^n$ in terms of $\gamma$ and the other derivatives $\dot\gamma^i$, $i=1,\ldots,n-1$: 
	\begin{equation}\label{gamman}
		\dot\gamma^n=\sqrt{\sum_{i,j=1}^{n-1}\hat g_{ij} \dot\gamma^i \dot\gamma^j}-\sum_{i=1}^{n-1}\hat \phi_i\dot\gamma^i.
	\end{equation}
	\item[\em Step 2]  Using the Lebesgue's number of the covering of $[a,b]$ obtained in {\em step 1}, we can choose  a partition of $[a,b]$,  $t_0=a<t_1<\ldots<t_{r-1}<t_r=b$ 
	with the property that $[t_i,t_{i+1}]$ is contained in one of the intervals of {\em step 1}.
	\item[\em Step 3] Given the first interval $[a,t_1]$ and the associated chart of {\em step 1}, which we denote now as $(A_1\times [a_1,b_1])$, 
	we can consider the projection $\hat\gamma_1$ of $\gamma|_{[a,t_1]}$ on $A_1$ and also the projection of the coordinates of the variational vector field $Z$, which we will denote by $\hat Z_1$. Then
	\[\hat\Lambda(t,w)=\hat\gamma_1(t)+w \hat Z_1(t)\]
	is a variation of $ \hat\gamma_1$ with variational vector field $\hat Z_1$ for $w\in(-\varepsilon,\varepsilon)$ with $\varepsilon>0$ small enough. Each curve $\hat\gamma_w(t)=\hat\Lambda(t,w)$ can be lifted to a lightlike curve by computing the $n$-coordinate  as the solution of   \eqref{gamman} with initial point $ \gamma^n(a)$. The lifted variation of lightlike curves has $Z$ as a variational vector field as both satisfy the differential equation \eqref{vareq} with initial condition $Z(a)=0$.
	\item[\em Step 4] 
	Proceed now with the second interval $[t_1,t_2]$ as in {\em step 3} in such a way that $\gamma|_{[t_1,t_2]}$ is contained in a chart $(\Omega_2,\psi_2)$ as in {\em step 1}  with $\Omega_2=A_2\times [a_2,b_2]$.  In this case, we can take as initial values to solve \eqref{gamman} the points $\gamma_w(t_1)$ obtained in {\em step 3}. We have a variation $\hat\Lambda:[a,t_2]\times(-\varepsilon,\varepsilon)\rightarrow M$ which is not smooth at $t_1$. Let us see how to smoothen it. The idea is to project the variation in a neighborhood of $t_1$  with the second chart.  Indeed, we can extend the variation of {\em step 3} to some $[a,t_1+\delta]$ and the one in the interval $[t_1,t_2]$ to $[t_1-\delta,t_2]$. Now we consider the projections of the variations on $A_2$  using $\pi_2:A_2\times[a_2,b_2]\rightarrow A_2$.  Consider two neighborhoods $U_1$ and $U_2$ of $ \hat\gamma_2(t_1)$ on $A_2$,  where $\hat\gamma_2$ is the projection of $\gamma|_{[t_1-\delta,t_2]}$ to $A_2$,  and consider a smaller $\delta$ if necessary in such a way that $\hat\gamma_2(t_1-\delta)\in U_1\setminus U_2$ and $\hat\gamma_2(t_1+\delta)\in  U_2\setminus U_1$.  We will assume that $Z$ is not tangent to $\gamma$ at $t_1$  (the tangent case will be considered in {\em step 4*}). Then in particular  $Z(t_1)$ is not zero and it does not belong to ${\rm span}\{\partial_n|_{\gamma(t_1)},\dot\gamma(t_1)\}$, which implies that the projection of $Z$ with $\pi_2$ is non-tangent to $\hat\gamma_2$ at $t_1$.  As a consequence, we have two lightlike variations of $\gamma|_{[t_1-\delta,t_1+\delta]}$ (obtained respectively with the first and the second chart), which can be projected with $\pi_2$ to variations in $A_2$. The fact that the projection of the variational vector field is not tangent to $\hat\gamma_2$ implies that  the  projected  variations generate a surface around $ \hat\gamma_2(t_1)$ and we can extend  its variational vector fields $\hat W_1$ and $\hat W_2$  to vector fields
	 $W_1$ and $ W_2$  in a neighborhood of $ \hat\gamma_2(t_1)$ (to obtain this we may need to choose smaller $\delta$ and neighborhoods $U_1$ and $U_2$). Using a partition of unity associated with $U_1$ and $U_2$ we can construct a vector field $ W$ that coincides with $ W_1$ in $U_1\setminus U_2$ and with $ W_2$ in $U_2\setminus U_1$.  As around the instant $t_1-\delta$ the restriction of $ W$ coincides with $ W_1$,  we can lift the variation  obtained by applying the flow of $W$ to $\hat\gamma_2|_{[t_1-\delta,t_1+\delta]}$  to lightlike curves using the initial conditions at $t_1-\delta$, which will make the lifted variation smooth with respect to that on $[a,t_1-\delta]$. Then we can lift again  the second projected variation on  the interval $[t_1+\delta,t_2]$ using the initial values at $t_1+\delta$ and will obtain the smooth variation by lightlike curves on $[a,t_2]$.
	\item[\em Step 4*]  Let $I$ be the subinterval of $[a,b]$ such that $\gamma|_I\subset \Omega_1\cap \Omega_2$.  If there is an instant $\bar t\in I$ where $Z$ is not tangent to $\gamma$, then, after redefining the partition, one may take $\bar t$ as the gluing time and proceed as in {\em step 4}.  When the vector field $Z$ is tangent to $\gamma$ along the whole interval $I$, it is not clear how to obtain the extended vector fields of {\em step 4}. In order to overcome this difficulty, 
 choose two neighborhoods $V_1$ and $V_2$ of $\gamma(t_1)$, both  contained in $\Omega_1\cap \Omega_2$, and $\delta>0$ small enough such that $\gamma(t_1-\delta)\in V_1\setminus V_2$ and $\gamma(t_1+\delta)\in V_2\setminus V_1$. Now let $X_1$ and $X_2$ be the $\partial_n$ in the first and the second chart, respectively. We use  a partition of unity associated with $V_1$ and $V_2$ to generate a vector field $X$ on $V_1\cap V_2$ from $X_1$ and $X_2$. Using $X$ and the hypersurface $x_n=0$ of the first chart we can generate a new chart which coincides with the first chart in a neighborhood of $\gamma(t_1-\delta)$ and it has the same $\partial_n$-integral curves as the second chart around $\gamma(t_1+\delta)$. Using this intermediate chart, we can obtain a variation as in {\em step 3} which matches smoothly with the variation in $[a,t_1-\delta]$ and can be continued in a smooth way after $t_1+\delta$ with the variation of {\em step 4}.  
	
\item[\em Step 5] We can proceed with the next intervals as in {\em step 4} until the last one. As we want to obtain lightlike curves that arrive to $\alpha$, we will include another interval $[t_r,b]$ as in {\em step 1} in such a way that in the chart $(\Omega_r,\psi_r)$ the curve $\alpha$ is obtained as an integral curve of $\partial_n$. This concludes the construction of the lightlike variation.
\end{enumerate}
\subsection{Reducing the problem to the classical calculus of variations}\label{theend}
One of the main problems to apply calculus of variations to the relativistic Fermat principle is that the arrival time is not computed in general as an integral along the lightlike curves. In order to overcome this problem, the proposal by V. Perlick  in \cite{Per90} was to express the variational vector fields $Z$ of lightlike curves as  $Z=W+f U$, where $W$ is a vector field along $\gamma$ such that $W(a)=W(b)=0$,  $f:[a,b]\rightarrow \R$ is a smooth function such that $f(a)=0$, and $U$ is the  parallel vector field along $\gamma$ such that $U(b)=\dot\alpha(s)$, being $s$ such that $\gamma(b)=\alpha(s)$. It turns out that $Z$ satisfies \eqref{vareq} if and only if
\begin{equation}\label{eqf}
	f(t)=-\int_a^t \frac{g(\frac{DW}{ds},\dot\gamma)}{g(U,\dot\gamma)} ds.
\end{equation}
Thus, let $Z$ satisfy \eqref{vareq} and let $\gamma_w$ be the variation defined by $Z$. By definition of $T$ we have
\[
\gamma_w(b)=\alpha(T(\gamma_w)).
\]
Differentiating at $w=0$ yields
\begin{equation}\label{Zb}
Z(b)=\frac{d}{dw}\Big|_{w=0} \gamma_w(b)
=\frac{d}{dw}\Big|_{w=0}\alpha(T(\gamma_w))
=\dot\alpha(T(\gamma))\,\frac{d}{dw}\Big|_{w=0}T(\gamma_w).
\end{equation}
Thus $\frac{d}{dw}|_{w=0}T(\gamma_w)$ is exactly the scalar coefficient $f(b)$,  in the decomposition
 $ Z(t)=W(t)+f(t)U(t)$, at $t=b$  (recall that $W(b)=0$ and $U(b)=\dot\alpha(T(\gamma))$). Hence,  $\gamma$ is critical if and only if  $dT(\gamma)[Z]=0$ for all variational vector field $Z$, i.e. if and only if $f(b)=0$ for all such $Z$.  From \eqref{eqf}, being $W$ any vector field along $\gamma$ with vanishing endpoints, 
 this turns out to be equivalent to
\[\frac{D}{dt}\left(\frac{1}{g(U,\dot\gamma)}\dot\gamma \right)=0,\]
which means that $\gamma$ is a pregeodesic.

\section{The regularity problem}
\subsection{The lack of regularity of the path space of time-oriented lightlike  curves in the $W^{1,p}$-topology}\label{W1p}

In this subsection we show that the path space of time oriented (future-pointing or past-pointing) lightlike curves is not a 
$C^1$ embedded submanifold of the full $W^{1,p}$-manifold of the  curves in the ambient space. We prove this in the simple case  when the ambient space is the Minkowski spacetime, so one in general  cannot expect  the  path space of time-oriented lightlike curves   to be a $C^1$-manifold in the $W^{1,p}$-topology. For the $H^1$-topology this was already observed with an example in \cite{GiaMa98}.

Let us denote by $I$ the interval $[0,1]$. Let $d\geq 2$, $p\in [1,\infty]$ and let $X:=W^{1,p}_0(I,\R^{1+d})=\{\gamma\in W^{1,p}(I,\R^{1+d}) : \gamma(0)=0\}$.
This is a closed Banach subspace of $W^{1,p}(I,\R^{1+d})$ where the only constant  curve is the one of constant value $0$.  Let  $\eta$ be the Minkowski metric  with signature $(-,+,\ldots,+)$  on $\R^{1+d}$ and let us define
\[
\mathcal{L}:=\{\gamma\in X:\ \eta(\dot\gamma,\dot\gamma)=0 \ \text{a.e. on }[0,1]\}.
\]
We use below the Bouligand tangent cone. For a set $S\subset X$ in a Banach space $X$ and $x_0\in \bar S$,
\[
C_{x_0}S:=\left\{v\in X:\ \exists\, x_n\in S,\ t_n\downarrow 0 \text{ with } x_n\to x_0 \text{ and } \frac{x_n-x_0}{t_n}\to v\right\}.
\]
If $S$ is a $C^{1}$ embedded submanifold of $X$, then $C_{x_0}S$ coincides with the usual tangent space $T_{x_0}S$,
hence is a \emph{linear} subspace of $X$.
\begin{prop}\label{singular}
Let  $\gamma_0\in\mathcal L$ and assume  that  there exists a measurable set $E\subset I$ with $|E|>0$ such that
$\dot\gamma_0=0$ a.e. on $E$.
Then $\mathcal L$ is not a $C^1$ embedded submanifold of $X$ in any neighborhood of $\gamma_0$. 
\end{prop}
\begin{proof}
If $\mathcal L$ were a $C^1$ embedded submanifold near $\gamma_0$, then $C_{\gamma_0}\mathcal L=T_{\gamma_0}\mathcal L$
and in particular $C_{\gamma_0}\mathcal L$ would be a \emph{linear} subspace of $X$. Let us then show that $C_{\gamma_0}\mathcal L$ is not linear.
	
Choose two null vectors $k_1,k_2\in\R^{1+d}$ which are not collinear
so that $\eta(k_i,k_i)=0$ but $\eta(k_1+k_2,k_1+k_2)<0$.
Define $v_i\in X$ by prescribing its derivative
\[
\dot v_i:=\chi_E\,k_i,\qquad v_i(0)=0.
\]
Let $t_n\downarrow0$ and set $\gamma_n^{(i)}:=\gamma_0+t_n v_i$.
Then
\[
\dot\gamma_n^{(i)}=\dot\gamma_0+t_n\chi_E k_i
=\begin{cases}
\dot\gamma_0&\text{on }I\setminus E,\\
t_n k_i&\text{on }E,
\end{cases}
\]
hence $\eta(\dot\gamma_n^{(i)},\dot\gamma_n^{(i)})=0$ a.e.  and, therefore, $\gamma_n^{(i)}\in\mathcal L$ for all $n$, and
\[
\frac{\gamma_n^{(i)}-\gamma_0}{t_n}=v_i\ \to\ v_i \quad\text{in }X.
\]
Thus $v_1,v_2\in C_{\gamma_0}\mathcal L$. Let us see now that $v_1+v_2\notin C_{\gamma_0}\mathcal L$.
Assume by contradiction that $v_1+v_2\in C_{\gamma_0}\mathcal L$. Then there exist $\gamma_n\in\mathcal L$ and
$t_n\downarrow0$ such that
\[
\frac{\gamma_n-\gamma_0}{t_n}\to v_1+v_2 \quad\text{in }X.
\]
Passing to derivatives  gives
\[
\frac{\dot\gamma_n-\dot\gamma_0}{t_n}\to \chi_E(k_1+k_2)\quad\text{in }L^p(I).
\]
Since $\dot\gamma_0=0$ a.e. on $E$, this implies
\[
\frac{\dot\gamma_n}{t_n}\to k_1+k_2\quad\text{in }L^p(E).
\]
If $p<\infty$,  $L^p(E)$ convergence yields (after passing to a subsequence) pointwise a.e. convergence on $E$.
If $p=\infty$, $L^\infty$ convergence also implies pointwise a.e. convergence.
Hence, 
\[
\frac{\dot\gamma_n}{t_n}(t)\to k_1+k_2 \quad\text{for a.e. }t\in E.
\]
But each $\dot\gamma_n$ is null a.e.  so for a.e. $t\in E$,
\[
\eta\!\left(\frac{\dot\gamma_n}{t_n},\frac{\dot\gamma_n}{t_n}\right)
=\frac{1}{t_n^2}\eta(\dot\gamma_n,\dot\gamma_n)=0.
\]
Passing to the a.e. limit yields $\eta(k_1+k_2,k_1+k_2)=0$, contradicting the choice that $k_1+k_2$ is timelike.
Therefore $v_1+v_2\notin C_{\gamma_0}\mathcal L$.
\end{proof}
Let us now  define for $p\in[1,\infty]$ the path space of future-pointing lightlike curves in $X$:
\[
\mathcal L^{1,p}_+:=\{\gamma\in X:\ \eta(\dot\gamma,\dot\gamma)=0\ \text{a.e. on }I,\ \dot t>0\ \text{a.e. on }I\}.
\]
(and the analogous space $\mathcal L^{1,p}_-$ of past-pointing lightlike curves).
Although every $\gamma\in \mathcal L^{1,p}_{\pm}$ satisfies $\dot\gamma\neq 0$ a.e., we can  prove that the sets
$\mathcal L^{1,p}_{\pm}$ accumulate (in the $W^{1,p}$-topology) on null curves in $\mathcal L$ whose velocity vanishes on a subset of \emph{positive} Lebesgue measure. Let us show this for the set $\mathcal L^{1,p}_{+}$, $ p\in [1,+\infty)$. 
\begin{lemma}\label{notlocclosed}
	 Let $\gamma_0\in \mathcal L^{1,p}_{+}$.  Then every neighborhood of $\gamma_0$ in $X$
contains a point $\bar\gamma\in X\setminus \mathcal L^{1,p}_+$
that is the limit of a sequence in $\mathcal L^{1,p}_+$ contained in the same neighborhood.
Equivalently, $\mathcal L^{1,p}_+\cap B_r(\gamma_0)$ is not closed in $B_r(\gamma_0)$ for any $r>0$.
Therefore $\mathcal L^{1,p}_+$ is not locally closed at $\gamma_0$.
\end{lemma}	
\begin{proof}
Let $r>0$.
Since $\dot\gamma_0\in L^p(I)$, by absolute continuity of the Lebesgue integral,  there exists a measurable
set $E\subset I$ with $|E|>0$ such that
\[
\|\dot\gamma_0\|_{L^p(E)}<\frac r2.
\]
Define a null curve $\bar\gamma=(\bar t,\bar x)\in X$ by prescribing its velocity:
\[
\dot{\bar\gamma}(s):=
\begin{cases}
	0 & s\in E,\\
	\dot\gamma_0(s)& s\in I\setminus E,
\end{cases}
\qquad \bar\gamma(0)=0.
\]
As  $\dot{\bar\gamma}=0$ a.e. on $E$, 
$\bar\gamma\notin\mathcal L^{1,p}_+$. Moreover (using on $X$ the equivalent $L^p$-norm for the derivative)
\[
\|\bar\gamma-\gamma_0\|_{W^{1,p}}
=\|\dot{\bar\gamma}-\dot\gamma_0\|_{L^p}
=\|\dot\gamma_0\|_{L^p(E)}<\frac r2,
\]
so $\bar\gamma\in B_r(\gamma_0)$.
Now define a sequence $\gamma_n=(t_n,x_n)\in X$ by
\[
\dot\gamma_n(s):=
\begin{cases}
	\frac1n\,\dot\gamma_0(s)& s\in E,\\
	\dot\gamma_0(s) & s\in I\setminus E,
\end{cases}
\qquad \gamma_n(0)=0.
\]
Then each $\gamma_n$ is null a.e.,  and satisfies
$\dot t_n>0$ a.e.
Hence $\gamma_n\in\mathcal L^{1,p}_+$ for all $n$.
Furthermore,
\[
\|\gamma_n-\bar\gamma\|_{W^{1,p}}
=\|\dot\gamma_n-\dot{\bar\gamma}\|_{L^p}
=\frac1n\,\|\dot\gamma_0\|_{L^p(E)}\to 0,
\]
so $\gamma_n\to\bar\gamma$ in $W^{1,p}$. In particular,
\[
\|\gamma_n-\gamma_0\|_{W^{1,p}}
\le \|\gamma_n-\bar\gamma\|_{W^{1,p}}+\|\bar\gamma-\gamma_0\|_{W^{1,p}}
<\frac r2+\frac r2=r,
\]
hence $\gamma_n\in B_r(\gamma_0)$ eventually.
\end{proof}
 In the case $p=\infty$, we have the following: 
\begin{lemma}\label{notlocclosedinfty}
	 Let $\gamma_0=(t_0,x_0)\in \mathcal L^{1,\infty}_{+}$ with   
\[
\essinf_{I}\dot t_0=0.
\]
Then every neighborhood of $\gamma_0$ in $X=W^{1,\infty}_0(I,\R^{1+d})$
contains a point $\bar\gamma\in X\setminus \mathcal L^{1,\infty}_+$
that is the limit of a sequence in $\mathcal L^{1,\infty}_+$ contained in the same neighborhood.
\end{lemma}
\begin{proof}
 Since $\essinf_{I}\dot t_0=0$,
for every $\varepsilon>0$ the measurable set
\[E_\varepsilon:=\{s\in I:\ 0<\dot t_0(s)<\varepsilon\}\] 
has  positive measure.
Since $\eta(\dot\gamma_0,\dot\gamma_0)=0$ and $\dot t_0>0$ a.e., we have
$| \dot x_0|=\dot t_0$ a.e., hence $|\dot\gamma_0|\le \sqrt2\,\dot t_0$ and
\[
\|\dot\gamma_0\|_{L^\infty(E_\varepsilon)}\le \sqrt2\,\varepsilon.
\]
Define $\bar\gamma\in X$ by prescribing its velocity:
\[
\dot{\bar\gamma}(s):=
\begin{cases}
	0& s\in E_\varepsilon,\\
	\dot\gamma_0(s)& s\in I\setminus E_\varepsilon,
\end{cases}
\qquad \bar\gamma(0)=0.
\]
Then $\eta(\dot{\bar\gamma},\dot{\bar\gamma})=0$ a.e., but $\dot{\bar t}=0$ on
$E_\varepsilon$, so $\bar\gamma\notin\mathcal L^{1,\infty}_+$.
Moreover,
$\|\dot{\bar\gamma}-\dot\gamma_0\|_{L^\infty}\le \sqrt2\,\varepsilon$ and 
\[
|\bar\gamma(s)-\gamma_0(s)|\le \int_I | \dot{\bar\gamma}-\dot \gamma_0|  =\int_{E_\varepsilon}|\dot\gamma_0|
\le |E_\varepsilon|\|\dot\gamma_0\|_{L^\infty(E_\varepsilon)}\le \sqrt2\,\varepsilon,
\]
hence $\|\bar\gamma-\gamma_0\|_{W^{1,\infty}}\le 2\sqrt2\,\varepsilon$.
Now define $\gamma_n\in X$ by
\[
\dot\gamma_n(s):=
\begin{cases}
	\frac1n\,\dot\gamma_0(s) & s\in E_\varepsilon,\\
	\dot\gamma_0(s)& s\in I\setminus E_\varepsilon,
\end{cases}
\qquad \gamma_n(0)=0.
\]
Then each $\gamma_n$ is null a.e.\ and satisfies $\dot t_n>0$ a.e., so
$\gamma_n\in\mathcal L^{1,\infty}_+$ for all $n$, and
\[
\|\gamma_n-\bar\gamma\|_{W^{1,\infty}}
= \|\dot\gamma_n-\dot{\bar\gamma}\|_{L^\infty}
+\|\gamma_n-\bar\gamma\|_{L^\infty}
\le \frac1n\|\dot\gamma_0\|_{L^\infty(E_\varepsilon)}+\frac1n\int_{E_\varepsilon}|\dot\gamma_0|
\longrightarrow 0,
\]
so $\gamma_n\to\bar\gamma$ in $W^{1,\infty}$. Choosing $\varepsilon>0$ so that
$2\sqrt2\,\varepsilon<r$, we get $\bar\gamma\in B_r(\gamma_0)$ and
$\gamma_n\in \mathcal L^{1,\infty}_+\cap B_r(\gamma_0)$ for $n$ large. 
\end{proof}
 As a consequence  of Lemmas~\ref{notlocclosed} and \ref{notlocclosedinfty} we have: 
\begin{prop}
The sets $\mathcal L^{1,p}_\pm$, $p\in[1,\infty]$, are  \emph{not} embedded $C^1$ Banach submanifolds of $X$.
\end{prop}
Indeed let us recall that if  $S\subset X$ is an embedded $C^1$ Banach submanifold and $x\in S$, then in a chart around $x$ the set $S$
is mapped to a closed complemented linear subspace of the model Banach space; hence $S\cap U$ is closed in $U$
for some neighborhood $U$ of $x$. In particular, $S$ is locally closed at each of its points (see \cite[Chapter 2, \S 2]{lang}).

\subsection{The regularity of the path space of $C^k$ lightlike future-pointing curves}\label{C1}
We consider now  the Banach manifold of $C^k$ curves ($k\in\mathbb N$, $k\ge1$), or  the Fr\'echet manifold of $C^\infty$ curves, and focus on the  future-directed (or past-directed) lightlike curves (null curves with  fixed
time orientation, hence with nowhere vanishing velocity). The argument below is identical to the one in \cite{GGP04}, with only notational changes and explicit bookkeeping of the $C^k$ (resp.\ $C^\infty$) framework.
\begin{rem}\label{nonempty}
	 Throughout the $C^k$--manifold discussion we assume that $\alpha$ is a timelike curve which is an embedding, i.e. $\mathrm{Im}(\alpha)$ is a one-dimensional submanifold of $M$.  Moreover we assume 
	\[
	\mathcal L^{k,\pm}_{p,\alpha}\neq\emptyset,
	\]
	i.e.\ there exists at least one $C^k$  lightlike curve
	$\gamma:[a,b]\to M$ such that $\gamma(a)=p$ and $\gamma(b)\in\mathrm{Im}(\alpha)$
	(with the prescribed time-orientation).
\end{rem}
Let $(M,g)$ be a time-oriented Lorentzian manifold of dimension $n$, $[a,b]\subset\R$, and fix $k\in\mathbb N$, $k\ge1$ (or $k=\infty$).
Fix a background Riemannian metric $h$ on $M$ with exponential map $\exp^h$.
Given a base curve $z_0\in C^k([a,b],M)$, consider the pullback bundle $z_0^*TM\to [a,b]$ and let
$\Gamma^k(z_0^*TM)$ be the Banach space (the Frech\'et space if $k=\infty)$ of $C^k$ sections $\xi:[a,b]\to TM$ along $z_0$.
For $\xi$ small enough  (so that $(z_0(t),\xi(t))$ stays in the domain of injectivity of $\exp^h$), define
\[
\Phi_{z_0}:\ \mathcal U_{z_0}\subset \Gamma^k(z_0^*TM)\longrightarrow C^k([a,b],M),\qquad
\Phi_{z_0}(\xi)(t):=\exp^h_{z_0(t)}(\xi(t)).
\]
Then $\Phi_{z_0}$ is a homeomorphism onto an open neighborhood of $z_0$ in the set $C^k([a,b],M)$; the collection of such maps
$\{\Phi_{z_0}\}_{z_0}$ provides an atlas, called the \emph{exponential atlas}.
Its transition maps are induced pointwise by smooth maps on $TM$, hence are smooth between the Banach spaces
of $C^k$ sections. In particular, $C^k([a,b],M)$ is a smooth Banach (Frech\'et, if $k=\infty$) manifold modeled on $C^k([a,b],\R^n)$.

Let $\alpha:I_\alpha\to M$ be a $C^k$ timelike curve and fix $p\in M$.
The endpoint evaluation map
\[
\mathrm{ev}_{a,b}:C^k([a,b],M)\to M\times M,\qquad \mathrm{ev}_{a,b}(z)=(z(a),z(b)),
\]
is smooth, and since $\alpha(I_\alpha)\subset M$ is a $1$--dimensional embedded submanifold of $M$,
the subset
\begin{align*}
C^k_{p,\alpha}([a,b],M):&=\{z\in C^k([a,b],M): z(a)=p,\ z(b)\in \alpha(I_\alpha)\}\\
&=\mathrm{ev}_{a,b}^{-1}\bigl(\{p\}\times \alpha(I_\alpha)\bigr)
\end{align*}
is a  smooth Banach submanifold of $C^k([a,b],M)$ (and a smooth Fr\'echet submanifold if $k=\infty$).

Inside it, the open \emph{regular} subset is
\[
C^{k,*}_{p,\alpha}([a,b],M):=\{z\in C^k_{p,\alpha}([a,b],M): \dot z(t)\neq 0\ \ \forall t\in[a,b]\}.
\]
We will shortly denote $C^{k}_{p,\alpha}([a,b],M)$ and $C^{k,*}_{p,\alpha}([a,b],M)$ by $C^{k}_{p,\alpha}$ and $C^{k,*}_{p,\alpha}$, respectively.

The (full) null locus and its regular part are
\[
\mathcal N^k_{p,\alpha}:=\{z\in C^k_{p,\alpha}: g\big(\dot z(t),\dot z(t)\big)=0\ \ \forall t\},
\qquad
\mathcal N^{k,*}_{p,\alpha}:=\mathcal N^k_{p,\alpha}\cap C^{k,*}_{p,\alpha}.
\]
Fix a smooth future-pointing timelike vector field $V_0$ on $M$ and define the time-oriented components
\[
\mathcal L^{k,+}_{p,\alpha}:=\{z\in \mathcal N^{k,*}_{p,\alpha}: g(\dot z,V_0)<0\},
\qquad
\mathcal L^{k,-}_{p,\alpha}:=\{z\in \mathcal N^{k,*}_{p,\alpha}: g(\dot z,V_0)>0\}.
\]
Proceeding as in \cite[Prop. 2.1]{GGP04}, we will prove the following. 
\begin{prop}\label{submani}
	With the above notation and $k\ge1$, $\mathcal N^{k,*}_{p,\alpha}$ is an embedded \emph{smooth} Banach
	submanifold of the open set $C^{k,*}_{p,\alpha}$. Moreover, $\mathcal L^{k,+}_{p,\alpha}$ and
	$\mathcal L^{k,-}_{p,\alpha}$ are open (and closed) submanifolds of $\mathcal N^{k,*}_{p,\alpha}$ and
	\[
	\mathcal N^{k,*}_{p,\alpha}=\mathcal L^{k,+}_{p,\alpha}\cup\,\mathcal L^{k,-}_{p,\alpha}.
	\]
	For $k=\infty$ the same holds in the Fr\'echet manifold $C^{\infty,*}_{p,\alpha}$.
\end{prop}

\begin{proof}
	On $C^{k,*}_{p,\alpha}$ consider
	\[
	\Psi_k:C^{k,*}_{p,\alpha}\longrightarrow C^{k-1}([a,b]),\qquad
	\Psi_k(z)(t):=\tfrac12\,g_{z(t)}(\dot z(t),\dot z(t)).
	\]
	In exponential charts, $z\mapsto (z,\dot z)$ is smooth from the $C^k$-topology to the $C^k\times C^{k-1}$-topology and
	$(x,v)\mapsto \tfrac12 g_x(v,v)$ is smooth on $TM$, hence $\Psi_k$ is smooth and
	\[
	\mathcal N^{k,*}_{p,\alpha}=\Psi_k^{-1}(0).
	\]
	Fix $z_0\in \mathcal N^{k,*}_{p,\alpha}$.
	A tangent vector at $z_0$ to $C^{k,*}_{p,\alpha}$ is a $C^k$ vector field $Z$ along $z_0$ such that
	\[
	Z(a)=0,\qquad Z(b)\in T_{z_0(b)}\alpha.
	\]
	 It is easy to show that $\Psi_k$  is differentiable and 
	\[
	d\Psi_k(z_0)[Z](t)=g\!\left(\dot z_0(t),\frac{D}{dt}Z(t)\right)\in C^{k-1}([a,b]),
	\]
	where $D/dt$ is the Levi--Civita covariant derivative along $z_0$.
	Let us choose a smooth timelike vector field $\widetilde Y$ on a neighborhood of $z_0([a,b])$ such that
	\[
	\widetilde Y(z_0(b))\in T_{z_0(b)}\alpha.
	\]
	Since $\dot z_0(t)$ is null and nonzero, $g(\dot z_0(t),\widetilde Y(z_0(t)))\neq0$ for all $t$; normalize along $z_0$ by
	\[
	Y(t):=\frac{\widetilde Y(z_0(t))}{g(\dot z_0(t),\widetilde Y(z_0(t)))}\qquad\Rightarrow\qquad
	g(\dot z_0(t),Y(t))\equiv 1,
	\]
	and still $Y(b)\in T_{z_0(b)}\alpha$.
	
	Given $h\in C^{k-1}([a,b])$, we consider  the scalar linear ODE for $\mu$ with initial condition
	\begin{equation}\label{ode}
	\dot\mu(t)+\mu(t)\,g\!\left(\dot z_0(t),\frac{D}{dt}Y(t)\right)=h(t),\qquad \mu(a)=0.
	\end{equation}
	The coefficients are $C^{k-1}$, hence there exists a unique solution $\mu\in C^k([a,b])$ depending continuously on $h$.
	Set $ Y_\mu :=\mu\,Y$. Then $ Y_\mu(a)=0$ by $\mu(a)=0$ and $ Y_\mu(b)=\mu(b)\,Y(b)\in T_{z_0(b)}\alpha$, so  $Y_\mu\in T_{z_0}C^{k,*}_{p,\alpha}$,  and
	\begin{align*}
	d\Psi_k(z_0)[ Y_\mu]
	&=g\!\left(\dot z_0,\frac{D}{dt}(\mu Y)\right)
	=\dot\mu\,g(\dot z_0,Y)+\mu\,g\!\left(\dot z_0,\frac{D}{dt}Y\right)\\
	&=\dot\mu+\mu\,g\!\left(\dot z_0,\frac{D}{dt}Y\right)
	=h.
	\end{align*}
	Thus $d\Psi_k(z_0)$ is surjective and admits a continuous linear right inverse; in particular, its kernel is complemented.  Indeed, the linear subspace of vector fields $\mu Y$ with $\mu$ an arbitrary $C^k$ function along $z_0$ is a complement, as the image of an arbitrary vector field  $Z$  coincides with the image of one of these vector fields (it is enough to set $h:=d\Psi_k(z_0)[Z]$ and determine  the unique $\mu$ that solves \eqref{ode} with that $h$, so that $Z-\mu Y\in \ker d\Psi_k(z_0)$). 
	By the Banach implicit function theorem, $\Psi_k^{-1}(0)=\mathcal N^{k,*}_{p,\alpha}$ is an embedded smooth
	Banach submanifold of $C^{k,*}_{p,\alpha}$.
	Finally, observing that the map
	\[
	H_k:C^{k,*}_{p,\alpha}\to C^{k-1}([a,b]),\qquad H_k(z)(t):=g\big(\dot z(t),V_0(z(t))\big)
	\]
	is continuous, on $\mathcal N^{k,*}_{p,\alpha}$ it never vanishes,  hence has constant sign on $[a,b]$. Therefore
	\[
	\mathcal L^{k,+}_{p,\alpha}=\mathcal N^{k,*}_{p,\alpha}\cap\{H_k<0\},\qquad
	\mathcal L^{k,-}_{p,\alpha}=\mathcal N^{k,*}_{p,\alpha}\cap\{H_k>0\}
	\]
	are open (and closed) in $\mathcal N^{k,*}_{p,\alpha}$, yielding the stated decomposition.

 For the smooth case ($k=\infty$), we notice that for each $k\ge 1$ we have Banach manifolds $C^{k,*}_{p,\alpha}$ with dense smooth inclusions
$i_{k+1,k}:C^{k+1,*}_{p,\alpha}\hookrightarrow C^{k,*}_{p,\alpha}$ and
$C^{\infty,*}_{p,\alpha}=\bigcap_{k\ge 1}C^{k,*}_{p,\alpha}$ (as a Fr\'echet manifold).
Moreover, the maps $\Psi_k$ satisfy the restriction property:
\[
\Psi_{k+1}=\Psi_k\circ i_{k+1,k},\qquad \Psi_k:C^{k,*}_{p,\alpha}\to C^{k-1}([a,b]).
\]
Fix $z_0\in\mathcal N^{\infty,*}_{p,\alpha}:=\Psi_\infty^{-1}(0)$.
In the proof of surjectivity of $d\Psi_k(z_0)$ one chooses a field $Y$ along $z_0$ with
$g(\dot z_0,Y)\equiv 1$ and $Y(b)\in T_{z_0(b)}\alpha$, and defines a right inverse
$R_k:C^{k-1}\to T_{z_0}C^{k,*}_{p,\alpha}$ by $R_k(h)=\mu(h)\,Y$, where $\mu(h)$ is the
solution of \eqref{ode}. If $h\in C^{k}$, then $\mu(h)\in C^{k+1}$, hence $R_{k+1}(h)=R_k(h)\in T_{z_0}C^{k+1,*}_{p,\alpha}$.
Consequently the splittings used in the Banach implicit function theorem can be chosen
compatibly with the inclusions $i_{k+1,k}$.
Applying the implicit function theorem at each level $k$ with these compatible splittings,
one obtains local submanifold charts for $\mathcal N^{k,*}_{p,\alpha}=\Psi_k^{-1}(0)$
that are compatible under restriction, i.e. the chart at level $k+1$ is the restriction
of the chart at level $k$ to the $C^{k+1}$-domain. Therefore the intersection
\[
\mathcal N^{\infty,*}_{p,\alpha}=\bigcap_{k\ge 1}\mathcal N^{k,*}_{p,\alpha}
\]
inherits a smooth Fr\'echet submanifold structure in $C^{\infty,*}_{p,\alpha}$.
\end{proof}	

\subsection{The Fermat principle in $\mathcal L^{k,\pm}_{p,\alpha}$}\label{Fermatsmooth}
 Having established the manifold structure for the sets $\mathcal L^{k,\pm}_{p,\alpha}$, $k\in (\mathbb N\setminus \{0\})\cup\{\infty\}$, we can obtain the proof of the Fermat principle as follows. 

 First notice that the arrival time functional in \eqref{T}, now defined on $\mathcal L^{k,\pm}_{p,\alpha}$, is smooth since it is the composition of the endpoint evaluation map 
\[
\mathrm{ev}_{b}:C^k([a,b],M)\to  M,\qquad \mathrm{ev}_{b}(z)=z(b),
\]
restricted to $\mathcal L^{k,\pm}_{p,\alpha}$, with the inverse map of the smooth timelike curve $\alpha$, i.e. $\alpha^{-1}:\mathrm{Im}(\alpha)\to I_\alpha$, which is smooth as well  (recall Remark~\ref{nonempty}). 

 Now, for each $\gamma\in \mathcal L^{k,\pm}_{p,\alpha}$ and $Z\in T_\gamma C^{k,*}_{p,\alpha}$, we have 
\[Z\in T_\gamma \mathcal L^{k,\pm}_{p,\alpha} \quad \Leftrightarrow\quad  d\Psi(\gamma)[Z]=0 \quad \Leftrightarrow \quad  g\left(\dot \gamma(t),\frac{D}{dt}Z(t)\right)=0,\]
i.e. $Z\in T_\gamma C^{k,*}_{p,\alpha}$ is also  tangent to $\mathcal L^{k,\pm}_{p,\alpha}$ if and only if \eqref{vareq} is satisfied.

 Thus, we can repeat the reasoning in Subsection~\ref{theend} if, for any  $C^k$ vector field $W$ along $\gamma$ with vanishing endpoints and for the  $C^k$ parallel vector field $U$ along $\gamma$ with  $U(b)=\dot\alpha(T(\gamma))$ we can find  a $C^k$ function $f_W:[a,b]\to \R$  with $f_W(a)=0$ such that $Z=W+f_WU$ belongs to $T_\gamma \mathcal L^{k,\pm}_{p,\alpha}$.
Since $U$ is timelike and $\dot\gamma$ is lightlike, one has
\[
g\bigl(U(t),\dot\gamma(t)\bigr)\neq 0,\qquad \forall\,t\in[a,b];
\]
then using $\tfrac{D}{dt}U=0$ we get 
\[
g\!\left(\dot\gamma,\frac{D}{dt}Z\right)
=
g\!\left(\dot\gamma,\frac{D}{dt}W\right)+\dot f_W\,g(U,\dot\gamma).
\]
Therefore \eqref{vareq} holds if and only if 
\begin{equation*}
	\dot f_W(t)=-\frac{g\!\left(\frac{D}{dt}W(t),\dot\gamma(t)\right)}{g\!\left(U(t),\dot\gamma(t)\right)},
	\qquad f_W(a)=0,
\end{equation*}
i.e., if and only if \eqref{eqf} is satisfied  (notice that since $W$, $U$ and $\gamma$ are $C^k$, $f_W$ is a $C^k$ function as well).   Thus,  if $\gamma\in \mathcal L^{k,\pm}_{p,\alpha}$ is a critical point of $T$ then 
\[f_W(b)=	-\int_a^b \frac{g(\frac{DW}{ds},\dot\gamma)}{g(U,\dot\gamma)} ds=0,\]
implying that $\gamma\in \mathcal L^{k,\pm}_{p,\alpha}$ must be a pregeodesic.

 The converse, i.e. any  $\gamma\in\mathcal L^{k,\pm}_{p,\alpha}$ which is a pregeodesic is a critical point of the arrival time functional $T:\mathcal L^{k,\pm}_{p,\alpha}\to \R$ defined in \eqref{T}, follows because,  as in the proof of (i) implies (ii) of \cite[Lemma 3]{Per90},  if $\gamma$ is a lightlike pregeodesic and $Z\in T_\gamma \mathcal L^{k,\pm}_{p,\alpha}$ then $Z(b)=0$. Thus,  from \eqref{Zb}, we conclude that $dT (\gamma)[Z]=0$.
\begin{rem}
The $C^k$ setting is well suited to
local bifurcation arguments as shown in  \cite{GGP04}.
On the other hand, for global variational existence results the $C^k$-topology is often less convenient, because it is
less compatible with the weak compactness  needed in infinite dimensions; in that context one
typically works in Sobolev spaces  where Palais--Smale condition and weak compactness  might be
available under some assumptions on the symmetry and/or the causal structure of the underlying spacetime that also allow us to circumvent the lack of the manifold structure on the path spaces of lightlike curves.
\end{rem}

\section{Bypassing the lack of the  structure of a Sobolev manifold: reductions and quadratic functionals.}\label{bypassing}
Even though the full path space of lightlike curves endowed with the $W^{1,p}$ (and in particular $H^1$) topology does not carry a  smooth manifold structure,  in several geometrically relevant settings one can still set up a workable variational theory by reducing the problem to  better functionals and/or ``reduced''  path spaces (see, e.g.,   the functional-analytic viewpoint in \cite{FGM95,GGM09}).

\subsection{Conformally stationary spacetimes}\label{conformal}  The class of conformally stationary spacetimes in the context of the Fermat principle attracted a wide  interest since the beginning. First recall that, as lightlike pregeodesics are invariant by conformal changes in the metric, we can assume that  the metric is stationary. As we said above, this case was already studied by Levi-Civita and Synge in the 1920's \cite{LeCi18,syn25}. And it was also considered by V. Perlick in a second part of his work about the general Fermat principle in spacetimes \cite{Per90ii}. 

The crucial point about stationary spacetimes is that  the arrival time can be computed up to a constant as the length of the projection into a hypersurface computed with a Finsler metric of Randers type.
More precisely, assume that $(M,g)$ admits a decomposition $M\simeq \Sigma\times\R$ with a distinguished time coordinate $t$ associated with the flow of the timelike Killing vector field and spacelike slices $\Sigma_t\simeq \Sigma$, then future-pointing lightlike curves $\gamma=(x,t)$ can be encoded by their spatial projections $x$ onto $\Sigma$. In fact, solving the null condition $g(\dot\gamma,\dot\gamma)=0$ with respect to $\dot t$ along future-pointing curves yields a Randers-type Finsler metric $F$ on $\Sigma$ (called {\em Fermat metric} in  \cite{CJM11},  see also \cite{CJS10}), and the arrival time becomes, up to constants, the $F$-length $\int F(x,\dot x)\,ds$ of the projected curve. This was already observed in \cite{MaPi98}.

Anyway, the length functional is invariant under reparametrization and therefore not coercive. 
To obtain a suitable functional, one can get inspiration from the case of Riemannian geodesics, where one considers the energy, namely, the integral of the square of the norm of  the velocity of the curve:
\[
E_F(x)=\frac12\int_0^1 F^2\bigl(x(s),\dot x(s)\bigr)\,ds,
\]
defined on a standard Hilbert manifold of curves on the slice  $\Sigma$ with fixed endpoints. After dealing with the limited regularity at $\dot x=0$ and ensuring the appropriate completeness properties, $E_F$ becomes the sharp functional for applying both Lusternik--Schnirelmann and Morse theories to lightlike geodesics in stationary spacetimes (see \cite{CJM11, CaJaMa10}). Actually $E_F$ is just $C^1$ in $H^1$-topology (see e.g. \cite{Cap10}) which is sufficient for Lusternik--Schnirelmann category arguments  aimed at multiplicity results.
This $C^{1}$-regularity is usually \emph{not} the right framework for the classical $C^{2}$-based Morse theory, but in the  stationary cases one can still recover Morse-type information by proving the Morse lemma  and the needed deformation lemma result (see \cite{CaJaMa10, CaJaMa13}).

\subsection{Stably causal spacetimes and quadratic Fermat principle}
The main problem with the arrival time functional is that in general it cannot be expressed as the integral of some function. But when the spacetime is stably causal, i.e. it admits  a globally defined temporal function, namely a smooth function $T:M\rightarrow \R$ with  timelike gradient (so strictly increasing along future-pointing w.r.t. the opposite  time orientation defined by $\nabla T$), we can parametrize the arrival curve using $T$ in such a way that $T(\alpha(s))=s$,  and the arrival  time function can be computed up to the addition of the constant $T(\gamma(a))$ as
\[ {\mathfrak L}(\gamma)=\int_a^b \frac{d}{ds} T(\gamma(s))ds.\]
Moreover, 
\[{\mathfrak L}(\gamma)=\int_a^b g(\dot\gamma(s),\nabla T_{\gamma(s)})ds.\]
As in the case of the energy functional $E_F$ above, the quadratic Fermat principle, first published in \cite{AP96}, considers the above functional taking the square of the integrand:
\[Q(\gamma)=\int_a^b g(\dot\gamma(s),\nabla T_{\gamma(s)})^2ds\]
and defined on one of the manifolds $\mathcal L^{2,\pm}_{p,\alpha}$ in Subsection~\ref{submani}.
\subsubsection{Proof of the quadratic Fermat principle}  The original Fermat principle can be very intuitive from the causality viewpoint. This is because at least when we consider lightlike geodesics which lie in the boundary of the causal future of the departing point, then it is clear that they must be the first curves to arrive into the observer. The situation with the quadratic Fermat principle is very different and the intuition provided by causality vanishes in this case. The proof was first accomplished in \cite{AP96} by F. Antonacci and P. Piccione and the most difficult part is to prove that if $\gamma$ is a critical point, then $g(\dot\gamma,\nabla T)$ is constant along $\gamma$. To prove that this is the case, we compute the differential of $Q$ in a variation by lightlike curves $\gamma_w$ with $w\in (-\varepsilon,\varepsilon)\subset \R$ and  $Z=\frac{d}{dw}\gamma_w|_{w=0}$. Using that $g(\dot\gamma,\nabla T)=\frac{d}{ds} T(\gamma(s))$ and $g(Z,\nabla T)=\frac{d}{dw} T(\gamma_w(s))$ and changing the order of the partial derivatives with respect to $s$ and $w$, one obtains that
\[dQ_\gamma(Z)=2\int_a^b g(\nabla T,\dot\gamma) \frac{d}{ds}g(\nabla T,Z)ds.\]
Then by applying integration by parts, as $Z(a)=0$,
\begin{equation}\label{dQZ}
	dQ_\gamma(Z)= 2g(\nabla T,\dot\gamma(b)) g(\nabla T,Z(b)) -2\int_a^b  g(\nabla T,Z)  \frac{d}{ds}g(\nabla T,\dot\gamma)ds.
	\end{equation}
Now consider $Z=W+f U$ as in Section \ref{theend}, namely, $W$ arbitrary with $W(a)=W(b)=0$, $U$ parallel along $\gamma$ with $U(b)$ tangent to the arrival curve $\alpha$ and 
\begin{equation*}
	f(t)=-\int_a^t \frac{g(\frac{DW}{ds},\dot\gamma)}{g(U,\dot\gamma)} ds.
\end{equation*}
By replacing $Z$ by $W+fU$, the first term to the right in \eqref{dQZ} can be expressed by applying integration by parts as $\int_a^b g(W, \frac{D}{ds}(\phi  \dot\gamma))ds$ for a certain function $\phi$, while the second term gives rise to the following two terms
\begin{multline*}
	-2\int_a^b  \frac{d}{ds}g(\nabla T,\dot\gamma) g(\nabla T,W)ds\\
	-2\int_a^b  \frac{d}{ds}g(\nabla T,\dot\gamma)\left( \int_a^s \frac{g(\frac{DW}{dr},\dot\gamma)}{g(U,\dot\gamma)} dr \right) g(\nabla T,U)ds.
	\end{multline*}
In the second term above, we can apply first a change in the order of integrals with respect to $r$ and $s$ and then integration by parts to obtain an expression of the type
\[\int_a^b g(W, \frac{D}{dr}(\psi \dot\gamma))dr.\]
Putting all together, we have  obtained  that 
\begin{multline*}\label{dQZ2}
	dQ_\gamma(Z)=\int_a^b g(W, \frac{D}{ds}(\phi \dot\gamma))ds-2\int_a^b g( \frac{d}{ds}g(\nabla T,\dot\gamma) \nabla T,W)ds  \\+\int_a^b g(W, \frac{D}{ds}(\psi \dot\gamma))ds.
	\end{multline*}
	It turns out that $\gamma$ satisfies that
	\begin{equation}\label{EL}
		\frac{D}{dt}(\phi \dot\gamma)-2\frac{d}{dt}g(\nabla T,\dot\gamma) \nabla T+\frac{D}{dt}(\psi \dot\gamma)=0.
		\end{equation}
	Now observe that 
	\begin{align*}
		g(\frac{D}{dt}(\phi \dot\gamma),\dot\gamma)&=\frac{d\phi}{dt} g(\dot\gamma,\dot\gamma)+\phi g(\frac{D}{dt}\dot\gamma,\dot\gamma)\\
 &=		\frac{d\phi }{dt} g(\dot\gamma,\dot\gamma)+ \frac 1 2 \phi\frac{d}{dt} g(\dot\gamma,\dot\gamma)=0,
		\end{align*}
		using that $\gamma$ is lightlike. Analogously, we have that $g(\frac{D}{dt}(\psi \dot\gamma),\dot\gamma)=0$, and therefore, multiplying \eqref{EL} by $\dot\gamma$ we conclude that
		\[-2\frac{d}{dt}g(\nabla T,\dot\gamma) g(\nabla T,\dot\gamma)=0,\]
		or equivalently
		\[\frac{d}{dt}g(\nabla T,\dot\gamma)^2=0, \]
		which implies that $g(\nabla T,\dot\gamma)$ is constant along $\gamma$, as we wanted to prove. Once we have proved that the integrand of $Q$ must be constant along critical points, it is straightforward to check that the Euler-Lagrange equations of $Q$ and $F$ are the same and, as a consequence, the critical curves of $Q$ are lightlike pregeodesics parametrized with $g(\nabla T,\dot\gamma)$ constant.
\subsubsection{Applications of the quadratic Fermat principle} The quadratic functional $Q$ has been used extensively to obtain multiplicity results and Morse-theoretic information for lightlike geodesics; see for instance \cite{GMP97,GMP98}. A common strategy is to work on suitable classes of timelike curves as approximating spaces, and then relate critical points indices to the lightlike limit.
However, as indicated by an example in the Appendix of \cite{GGM09},  in Sobolev frameworks, one still faces  regularity obstructions: the natural sets of future-pointing  curves defined by timelike constraints need not carry a Hilbert manifold structure  and the associated functionals may fail to have the differentiability properties required for a direct application of the classical infinite-dimensional Morse theory, as already the stationary case described in Subsection~\ref{conformal} clearly shows.

A complementary route goes back to Uhlenbeck \cite{Uhlenb75}, who developed a Morse theory for Lorentzian geodesics on globally hyperbolic Lorentzian  manifolds   using piecewise-smooth curves, thus avoiding the need for a  manifold structure on the constrained path space. In the context of gravitational lensing, this philosophy is exploited effectively in \cite{GMP02}, where a quite general version of the Fermat principle is proved on the set of $W^{1,\infty}(I,M)$ lightlike curves connecting a point to a timelike curve $\alpha$ and whose velocity norm (with respect to an auxiliary Riemannian metric) has essential infimum bounded  away from zero (see \cite[Th.~3.9]{GMP02}) (recall also  Lemma~\ref{notlocclosedinfty}).  Since the authors do not assume stable causality, their Fermat principle is closer in spirit to Section~\ref{Fermat}. This Fermat principle is then combined with a curve-shortening scheme to derive Morse relations for lightlike geodesics joining an event to a smooth timelike curve (also allowing for the presence of a boundary), under a compactness assumption for the arrival time functional called \emph{pseudocoercivity}, introduced in the study of geodesic connectedness of stationary Lorentzian manifolds in \cite{GP99}.

\section{Extensions of the Fermat principle}

\subsection{Extended sources and receivers}
There have been quite a few extensions of the Fermat principle in GR, but the first one was published by V. Perlick and P. Piccione in \cite{PP98}. To consider an extended source is something quite natural in the same way as one considers the problem of minimizing the distance between one submanifold and a point in a Riemannian manifold. The question of considering an extended receiver is more subtle as one needs to define an arrival time function  to state the Fermat principle. This problem was solved in \cite{PP98} by considering as a receiver a timelike submanifold endowed with a temporal function. Then the evaluation of the temporal function in the arrival point of the lightlike curve $\gamma$ defines the arrival time to establish the Fermat principle.

\subsection{Timelike geodesics with a prescribed proper time}
Lightlike geodesics can be viewed as causal geodesics with zero proper time. This naturally leads to the question: what characterizes causal geodesics with a fixed proper time  $T_0>0$?
Is it possible to generalize the Fermat principle to this case? The answer is affirmative as   shown first in \cite{GMP98b}  and it can be easily recovered using the classical Fermat principle for lightlike geodesics in a spacetime of higher dimensions   as done in \cite[\S 4.3]{CJM11}.  Given a Lorentzian manifold $(M,g)$, one can apply the Fermat principle to $(M\times \R,g+dr^2)$. We consider  lightlike curves  between an event $(p,0)$ and a  timelike curve $[0,1]\ni s\rightarrow (\alpha(s), T_0)$. Then a lightlike curve $[0,1]\ni s\rightarrow \beta(s)=(\gamma(s), r(s))\in M\times \R$ satisfies that $g(\dot\gamma,\dot\gamma)+\dot r^2=0$.  If $\gamma$ is timelike, then $\dot r$ cannot be zero, and as $r(0)=0$ and $r(1)=T_0>0$, we deduce that $\dot r>0$.  
Furthermore,
\[\ell_g(\gamma)=\int_0^1\sqrt{-g(\dot\gamma(s),\dot\gamma(s))}ds =\int_0^1 \dot r(s)ds = T_0,\]
in other words, the proper time of $\gamma$ is $T_0$.  Moreover, observe that as $\partial_r$ is a Killing field, then $r(s)$ is an affine function for the  geodesics of $(M\times \R,g+dr^2)$.  Thus,  a lightlike geodesic $s\rightarrow (\gamma(s),r(s))$ from $(0,p)$ to the curve $(\alpha(s),T_0)$ must have $\gamma$ timelike with $\dot r>0$ constant.  As a consequence, the Fermat principle for lightlike geodesics in $(M\times \R,g+dr^2)$ between an event $(p,0)$ and a timelike curve
$[0,1]\ni s\rightarrow (\alpha(s), T_0)$  gives rise to a Fermat principle for timelike geodesics with prescribed proper time, namely, the timelike geodesics with fixed proper time $T_0$ are the critical points of the arrival time function between an event $p$ and a timelike curve $\alpha$ between all the timelike curves with fixed proper time $T_0$. 

\subsection{Arbitrary arrival curve} Traditionally,  the Fermat principle has been formulated by taking a timelike curve as the arrival set. This makes sense because in that case the curve can be interpreted as an observer and the arrival time as the proper time of the observer. However,  there are situations in which it is natural to consider spacelike curves for the arrival set,  for example in  spacetimes endowed with a Killing field  $K$ that is not necessarily timelike. In such a case, the orbits of the Killing field are natural candidates for the arrival curve,  and when  $K$ is lightlike or spacelike the classical Fermat principle does not apply. This setup motivated the generalization of the classical Fermat principle to arbitrary arrival curves in \cite{CJS14}. In this case,  one must exclude the situation in which the lightlike geodesic is orthogonal to the arrival curve. This cannot happen when the arrival curve is timelike since  timelike vectors cannot be orthogonal to lightlike ones in a Lorentzian manifold. If the  lightlike geodesic is not orthogonal to the arrival curve, then it is a critical point of the arrival time among all the lightlike curves.  We also must observe that in this case the term ``arrival time'' cannot be interpreted as the proper time of  the arrival curve   since  proper time is undefined for non-timelike curves. Moreover, in  the case of a Lorentzian manifold with a Killing field, lightlike geodesics satisfy the Fermat principle considering as arrival curve the orbit of the Killing field even when the lightlike geodesic is orthogonal to the orbit (see \cite[Th. 7.4]{CJS14}).

\subsection{Finsler spacetimes and cone structures} From a geometric viewpoint, Finsler spacetimes are one of the most natural generalizations of classical Lorentzian spacetimes, adding the possibility of local anisotropy in the spacetime.  Taking into account that causal notions and geodesics can be directly extended to a Finsler setup, the Fermat principle can be formulated word for word in Finsler spacetimes. The proof follows the same steps as the classical one, with some additional technicalities, and it was first published  in \cite{Per08}. 
An extension of the Fermat principle in stationary spacetimes described in Subsection~\ref{conformal} to an analogous class of Finsler spacetimes   was obtained in \cite{CapSta18}.
A formulation 
for  both timelike and lightlike curves in 
time-oriented Finsler spacetimes was established in \cite{GPV12}, working on manifolds analogous to $\mathcal L^{\infty,+}_{p,\alpha}$,  but  encompassing the possibility of a fixed timelike  constraint  on a given interval of parametrization.    In \cite{GPV12} the authors also compute the second variation of the arrival time functional and state an extension of the  Morse Index Theorem for timelike geodesics from 
Lorentzian geometry to the Finslerian framework.

  As for   the extension to (strong) cone structures (see \cite[Def. 2.7]{JS20}), it is straightforward,   since   any cone structure admits a Lorentz-Finsler metric  with the cone structures as its lightlike cones (see \cite[Cor. 5.8]{JS20}) and lightlike geodesics are determined for the cone structure up to reparametrization (see \cite[Th. 6.6]{JS20}).

\subsection{Snell law in Finsler spacetimes}  As we have said above, light trajectories are determined by GR in smooth settings. This was one of the departing points to formulate the theory of general relativity. But nothing was said about the case in that the metric of the spacetime has a discontinuity. It is reasonable to assume that light rays are given by a Fermat principle, which in a certain sense it is equivalent to the generalized Huygen's principle. This was studied in \cite{JMPS25} for Finsler spacetimes obtaining a generalization of Snell law.  For simplicity, let us state the result for classical spacetimes. Assume that we have a manifold $M$ with a hypersurface $\eta$, the interface, that divides $M$ into two connected parts $M_1$ and $M_2$, and in each one of these connected parts we have respectively, two  Lorentzian metrics $g_1$ and $g_2$ which extend to $\eta$ but do not match (there is a discontinuity).  Then if $\gamma_1:[a,t_0]\rightarrow M_1$ is an incident light ray from $(M_1,g_1)$ (a lightlike geodesic) that intersects $\eta$ at the instant $t_0\in I$, 
the refracting trajectory is a lightlike geodesic $\gamma_2:[t_0,b]\rightarrow M_2$ of $(M_2,g_2)$ with $\gamma_1(t_0)=\gamma_2(t_0)=q$ and 
 \[\dot\gamma_2^{\perp_{g_2}}\cap T_{q}\eta=\dot\gamma_1^{\perp_{g_1}}\cap T_{q}\eta.\]
 Even if this expression does not look like the classical Snell law, using a static spacetime, one can recover the classical Snell law from it (see \cite[Example 7.5]{JMPS25}).

\section*{Acknowledgments}

\noindent  E.Caponio is  partially supported   by  MUR under the Programme ``Department of Excellence'' Legge 232/2016  (Grant No. CUP - D93C23000100001) and by ``INdAM - GNAMPA Project''  CUP E53C25002010001.\\
M.A. Javaloyes is  partially supported  by    Project PID2021-124157NB-I00, funded by MCIN/AEI/10.13039/501100011033/ and ``ERDF A way of making Europe"; and by ``Ayudas a proyectos para el desarrollo de investigaci\'{o}n cient\'{i}fica y t\'{e}cnica por grupos competitivos (Comunidad Aut\'{o}noma de la Regi\'{o}n de Murcia)'', included in ``Programa Regional de Fomento de la Investigaci\'{o}n Cient\'{i}fica y T\'{e}cnica (Plan de Actuaci\'{o}n 2022)'' of Fundaci\'{o}n S\'{e}neca-Agencia de Ciencia y Tecnolog\'{i}a de la Regi\'{o}n de Murcia, REF. 21899/PI/22.

%

%

\end{document}